\documentclass{article}

\usepackage{./arxiv}

\usepackage{soul}
\usepackage{url}
\usepackage{hyperref}
\usepackage[utf8]{inputenc}
\usepackage[T1]{fontenc}
\usepackage[small]{caption}
\usepackage{amsmath}
\usepackage[capitalise]{cleveref}
\usepackage{microtype}      
\usepackage{amssymb}
\usepackage{amsthm}
\usepackage{thm-restate}
\usepackage{bm}
\usepackage[english,provide=*]{babel}
\usepackage{graphicx}
\usepackage{natbib}
\usepackage{doi}

\newtheorem{theorem}{Theorem}
\newtheorem{lemma}{Lemma}
\newtheorem{corollary}{Corollary}
\newtheorem{remark}{Remark}
\newtheorem{observation}{Observation}
\newtheorem{hypothesis}{Hypothesis}
\theoremstyle{definition}
\newtheorem{definition}{Definition}
\theoremstyle{remark}

\crefname{observation}{Observation}{Observations}
\crefname{hypothesis}{Hypothesis}{Hypotheses}

\title{Improved Approximation Ratio for Strategyproof Facility Location on a Cycle}

\newif\ifuniqueAffiliation
\uniqueAffiliationtrue

\ifuniqueAffiliation 
\author{ \href{https://orcid.org/0009-0006-7311-5628}{\includegraphics[scale=0.06]{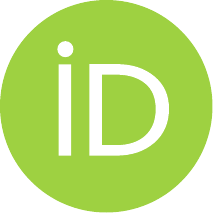}\hspace{1mm}Krzysztof Rogowski} \\ 
	University of Warsaw\\
	Institute of Informatics\\
	\texttt{kr418382@mimuw.edu.pl} \\
	\And
	\href{https://orcid.org/0000-0003-1756-2424}{\includegraphics[scale=0.06]{orcid.pdf}\hspace{1mm}Marcin Dziubiński} \\
	University of Warsaw\\
	Institute of Informatics\\
	\texttt{m.dziubinski@mimuw.edu.pl} \\
}
\else
\fi

\renewcommand{\shorttitle}{Improved Approximation Ratio for Strategyproof Facility Location on a Cycle}

\hypersetup{
pdftitle={Improved Approximation Ratio for Strategyproof Facility Location on a Cycle},
pdfauthor={Krzysztof Rogowski, Marcin Dziubiński},
pdfkeywords={Facility Location, Strategyproof Mechanisms, Social Cost Minimization, Cycle Graph},
}

\DeclareMathOperator{\Ima}{Im}

\begin{document}
\newcommand{\TODO}[1]{\textcolor{red}{todo: #1}}
\newcommand{\cc}[2]{\left[ #1, #2 \right]} 
\newcommand{\co}[2]{\left[ #1, #2 \right)} 
\newcommand{\oc}[2]{\left( #1, #2 \right]} 
\newcommand{\oo}[2]{\left( #1, #2 \right)} 
\newcommand{\set}[1]{\left\{ #1 \right\}} 
\newcommand{\quo}[1]{``#1''}
\newcommand{\ifrac}[2]{#1/#2 }
\newcommand{\apx}{\text{apx}}
\newcommand{\soc}{\text{sc}}
\newcommand{\cost}{\text{c}}
\newcommand{\gl}[1]{\overline{#1}}
\newcommand{\opt}{\text{opt}}
\maketitle

\begin{abstract}
    We study the problem of design of strategyproof in expectation (SP) mechanisms for
    facility location on a cycle, with the objective of minimizing the sum
    of costs of $n$ agents. We show that there exists an SP mechanism that
    attains an approximation ratio of $7/4$ with respect to the sum of costs
    of the agents, thus improving the best known upper bound of
    $2-2/n$ in the cases of $n \geq 5$. The mechanism obtaining the
    bound randomizes between two mechanisms known in the literature: the
    Random Dictator (RD) and the Proportional Circle Distance (PCD)
    mechanism of~\citet{meir2019strategyproof}. To prove the result, we propose a
    cycle-cutting technique that allows for estimating the problem on a~cycle by a~problem on a~line.
\end{abstract}

\keywords{Facility Location \and Strategyproof Mechanisms \and Social Cost Minimization \and Cycle Graph}

\section{Introduction}
\label{sec:intro}
The facility location problem involves a group of $n$ agents, each having a preference over a set of locations within a~metric space. The agents have their own ideal location in the space and they prefer locations that are closer to their ideal location. A central authority (the social planner), not knowing the ideal points of the agents, has to choose one of the locations, aiming to minimize the sum of distances between the locations of the agents and the chosen location (the so-called utilitarian welfare objective). An additional problem, beyond the choice of optimal location, arises due to the lack of knowledge of the ideal points of the agents. To address this problem the social planner asks the agents to report their location and uses a mechanism which, given the reports of the agents, determines the location to be chosen. Since the agents will report their location in order to minimize their own distance to the location chosen by the mechanism, there is no guarantee that their reports will be truthful. To resolve this problem, the social planner is restricted to choosing a~strategyproof mechanism, under which reporting true ideal points is individually weakly optimal, regardless of the reports of other agents. The outcomes of the mechanisms may be either deterministic or randomized.

The strategyproof facility location problem attracted interest from researchers for nearly 50 years now and it actively researched to this day. This is due to natural applications, like locating public facilities (schools, healthcare facilities, etc.) in towns where each of the citizens has own preference regarding the best location. These natural, physical, applications extend to virtual ones, like choosing the best time for a~meeting.

In this paper we are interested in strategyproof facility location problems on graphs, specifically on cycles.
In the seminal work, \citet{moulin1980strategy} obtained a complete characterization of strategyproof mechanisms when the space of possible locations is a line segment. He showed, in particular, that the mechanism choosing the median of the reported points is not only strategyproof but also efficient, in the sense that it minimizes the sum of distances to the ideal points of the agents. \citet{schummer2002strategy} extended the characterization of strategyproof mechanisms to graphs. It follows from their characterization that when the graph is a tree then there exists a strategyproof mechanism which is also efficient. If a graph contains a cycle, however, no deterministic strategyproof mechanism is efficient. In a later work, \citet{meir2019strategyproof} showed that this is also true for randomized mechanisms.
This raises a question: how close to efficiency can we get when cycles are present?

To address possible inefficiency of strategy proof mechanisms, \citet{procaccia2013approximate} proposed the idea of approximate mechanism design (without money). They introduced the approximation ratio which measures how good the best outcome of a strategy proof mechanism is, related to the optimal outcome. The best known upper bound on the approximation ratio of strategyproof mechanisms for facility location on a graph (under utilitarian welfare objective), $2- n/2$, was obtained by~\citet{alon2009strategyproof} with use of a random dictator (RD) mechanism. \citet{meir2019strategyproof} improved this bound for the case of a cycle with $n = 3$ agents to $7/6$. In this paper we are interested in the approximation ratio of the strategyproof facility location on a cycle.

\subsection{Related Work}

The literature on facility location, in general, and on strategyproof facility location, in particular, is vast. In the interest of space, in this review we restrict attention to the paper that are closest related to our work. For an excellent recent literature review of this field see~\citet{chan2021mechanism}.

Our work falls into the area of mechanism design without money as applied to facility location. Fundamental papers in this area, \citet{moulin1980strategy}, \citet{schummer2002strategy}, and~\citet{procaccia2013approximate} were already discussed above.

We are specifically interested in the approximation ratio of strategy proof facility location mechanisms on a cycle, with a~single location being selected and under utilitarian welfare objective. As already mentioned, \citet{alon2009strategyproof} obtained the upper bound of $2-2/n$ on this approximation ratio, by the RD mechanism. This result was improved by~\citet{meir2019strategyproof} for the case of $n=3$ agents. He defines new strategy proof mechanisms for a cycle, the Proportional Circle Distance (PCD) mechanism and the $q$-Quadratic Circle Distance ($q$-QCD) mechanism. The $1/4$-QCD is strategyproof in the case of $n=3$ agents and obtains the tight bound of $7/6$. The PCD mechanism is strategyproof for any odd number of agents. However, when the number of agents grows, its approximation ratio approaches $2$.
\citet{dokow2012mechanism} studied the strategyproof facility location on graphs where only vertices of the graph can be chosen. In the case of cycles they showed that when the number of vertices is sufficiently large, the strategyproof mechanisms must be close to dictatorial. In the case of small numbers of vertices, non-dictatorial, anonymous, strategyproof mechanisms exist.
Other notable works on strategyproof facility locations on cycles, with objective other than utilitarian social welfare, include~\citet{alon2010walking} and~\citet{cai2016facility}.

\subsection{Our contribution}

We obtain a new upper bound on the approximation ratio for the strategyproof facility location on a cycle. We show that this ratio is bounded from above by $7/4$ for any odd $n \geq 3$. This improves the previously known upper bound of $2-2/n$ for the case of $n\geq 5$.
To obtain this result, we propose a~mixed randomized mechanism, RD+PCD, which mixes between using the RD mechanism and the PCD mechanism. The key challenge in proving the bound are non-linear distances on the cycle. To overcome this issue we propose a cut technique, which ``cuts'' the cycle in a properly chosen point and allows for reducing the analysis of the bounds on a cycle to a line segment. We complement the theoretical analysis of the bound with computational analysis of the performance of the RD+PCD mechanism. The analysis suggests that the upper bound could be improved even further, to $3/2$. It is important to note that in the paper we consider the notion of strategyproofness in expectation, weaker to the notion of universal strategyproofness~\citep{mu2018setting}.

\section{Preliminaries}

A facility location problem consists of a set of agents \( N \) and a domain, typically a metric space. Throughout this paper, we assume that \( |N| \geq 3 \) is an odd number and, for convenience, we define \( k = (|N| - 1)/2 \). We index the agents symmetrically around zero, so that \( N = \{-k, -k+1, \ldots, -1, 0, 1, \ldots, k\} \).

The domain is usually taken to be a graph equipped with a metric. In this work, we focus on a continuous cycle of length \( 1 \). Representing this cycle as a subset of \( \mathbb{R} \) simplifies the notation used in the subsequent sections.

A \emph{cycle} of length \( 1 \), obtained from the line segment \( \cc{-0.5}{0.5} \) by joining its endpoints,
is denoted by \( G \).
For the purposes like comparisons or arithmetic operations,
the joined endpoint of the cycle is \( -0.5 \).
The distance between any two points \( v_1, v_2 \in G \) on the cycle is given by
\( d(v_1, v_2) = \min(|v_2 - v_1|, 1 - |v_2 - v_1|) \).
Let us orientate the cycle in such a way that clockwise movement corresponds to increasing point values with one exception of passing through the joined endpoint \( -0.5 \).
An \emph{arc} between two points \( v_1, v_2 \in G \) on the cycle is defined as the set of points traversed on the cycle when moving clockwise from \( v_1 \) to \( v_2 \).

Every agent \(i \in N\) has preferences over the points of the graph, determined by their ideal point, \( b_i \). These preferences are determined by the \emph{cost function}, defined as the distance between the agent's ideal point \( b_i \) and a chosen point \( v \in G \). When comparing two points on the graph, agent \( i \) prefers the point with lower cost.

The collection of agents' ideal points constitutes a \emph{profile} and is represented by a function $b : N \rightarrow G$, where \( b_i \) is the ideal point of agent \( i \) and $b^{-1}(v)$ are the agents reporting $v$.

For any set \( A \), a \emph{lottery} over \( A \) is defined as a discrete probability distribution on \( A \).
We represent a lottery as a function \( l: A \to \cc{0}{1} \) satisfying the condition
that the sum of values in its image equals \( 1 \).
The set of all lotteries over \( A \) is denoted by \( \Delta(A) \).

Agents' preferences are extended to lotteries over the points of the graph by using the expected values. The \emph{cost} for an~agent with ideal point \( v \) under a lottery \( l \in \Delta(G) \) is a~linear extension of the deterministic cost, defined as the expected distance between \( v \) and the outcome of the lottery:
\[ c_v(l) = \mathbb{E}_{X \sim l}[d(v, X)]. \]
An agent prefers lotteries with lower costs.

In this work, the quality of any lottery \( l \in \Delta(G) \) is evaluated using the \emph{social cost},
defined as the sum of costs of all agents.
Formally, for any profile \( b \), let
\begin{equation*}
    \soc_b(l) = \sum_{i \in N} c_{b_i}(l).
\end{equation*}

For simplicity, we sometimes abuse the notation and compute the social cost for a single point \( v \in G \),
treating it as a~degenerate lottery where the given point is always selected.
The minimal social cost is referred to as the \emph{optimal cost} and denoted by \( opt_b \).
Formally,
\begin{equation*}
    \opt_b = \inf_{v \in G} \soc_b(v).
\end{equation*}
It is straightforward to see that, due to the definition of agent costs via expected values,
the optimal cost is always achieved by some point.

In this paper, we consider the problem faced by a social planner who, without knowledge of the true profile, aims to choose a lottery over the points of the graph to minimize social cost. To achieve this, the social planner employs a~mechanism design approach.

\begin{definition}\label{def:mechanizm}
    Let \( B \) be an arbitrary set of profiles.
    A function \( M: B \to \Delta(G) \), which maps profiles to lotteries over the points of the graph,
    is called a \emph{mechanism}.
\end{definition}

Given a mechanism, agents report their ideal points, not necessarily truthfully. Based on these reports, the mechanism determines a lottery, which is then used to select a point of the graph.

\subsection{Common Mechanism Properties}

Let \( b \) denote an arbitrary profile.

\begin{definition}
    A mechanism \( M \) is \emph{strategyproof (SP)} if, for any agent \( i \in N \), the cost incurred by \( i \) when truthfully reporting their ideal location \( b_i \) is no greater than the cost incurred when reporting any other point \( v \in G \). Formally:
    \begin{equation*}
        c_{b_i}(M(b)) \leq c_{b_i}(M(b[i \to v])),
    \end{equation*}
    where \( b[i \to v] \) denotes the profile in which agent \( i \) reports point \( v \) instead of their true ideal location \( b_i \).
\end{definition}

Strategyproofness implies that no agent can gain by misreporting their ideal location to the mechanism. Consequently, when considering SP mechanisms, and under the assumption that agents act rationally, they will truthfully report their ideal points.

In this work, we focus exclusively on SP mechanisms. Based on the above reasoning, we make little distinction between the true profiles of agents' ideal points and the profile reported by the agents, assuming they are equal.

In addition, mechanism \( M \) is said to be:
\begin{itemize}
    \item \emph{Anonymous}, if it does not distinguish between agents, that is for any permutation of agents \( \pi: N \to N \),
          the outcome of the mechanism \( M \) remains unchanged.
          Formally, \( M(b) = M(b \circ \pi) \).

    \item \emph{Neutral}, if it does not distinguish between similar graph's locations. Formally, for any automorphism \( f \) of the graph \( G \),
          the outcome of the mechanism \( M \) is transformed according to \( f \) i.e., \( M(f \circ b)\circ f = M(b) \).

    \item \emph{Peaks-only}, if the support of the resulting lottery \( M(b) \) is contained
          within the points reported by the agents, i.e., in the image \( \Ima b \).

    \item \emph{\( \alpha \)-approximation} if the social cost of the mechanism's outcome is
          no greater than \( \alpha \) times the optimal cost.
          Formally, \( \soc_b(M(b)) \leq \alpha \cdot opt_b \).
\end{itemize}

The \emph{approximation ratio} of a mechanism \( M \) for a set of profiles \( B \) is defined as the smallest \( \alpha \) such that \( M \) is an \( \alpha \)-approximation for all profiles from \( B \).
We denote this value by \( \apx_M(B) \).
We often abuse the notation and write \( \apx_M(b) \), where \( b \) is a single profile,
to denote the approximation ratio of the mechanism \( M \) for the set \( \set{b} \).
In this paper we are interested in mechanisms with the smallest approximation ratio for all profiles on the cycle.
This is equivalent to finding the upper bound of \( \apx_M \) for singletons that holds uniformly across all profiles.

\subsection{Known SP mechanisms}
We now present two important mechanisms from the literature that are central to our work.
\begin{definition}
    \emph{RD} (Random Dictator) is a mechanism that returns
    a lottery \( l \in \Delta(G) \), where the probability of selecting any point \( v \in G \) is
    proportional to the number of agents choosing it in profile \( b \). Formally, \( l(v) = |b^{-1}({v})| / |N| \).
\end{definition}
It is known~\citep{alon2009strategyproof} that for any graph and any profile, the approximation ratio of this mechanism does not exceed \( 2 - 2/|N| \).

\begin{definition}
    The \emph{Proportional Circle Distance (PCD)} mechanism~\citep{meir2019strategyproof} is defined for cycles and an odd number of agents. The mechanism operates according to following steps:
    \begin{enumerate}
        \item Fix a linear order \( \prec \) of agents corresponding to the clockwise arrangement of their reports on the cycle, with ties broken arbitrarily (it can be verified that this does not affect the outcome of the mechanism).
        \item For each agent \( i \), define the \emph{arc opposing} agent \( i \) as the arc between the report of the agent who is \( k \) positions after \( i \) in the order \( \prec \) and the report of the agent who is \( k \) positions before \( i \) in the order \( \prec \). The terms \quo{after} and \quo{before} are understood with respect to the clockwise traversal of the cycle; for instance, the agent immediately after the last agent in \( \prec \) is the first agent in \( \prec \).
        \item Select the report of each agent \( i \) with a probability proportional to the length of the arc opposing agent \( i \).
    \end{enumerate}
\end{definition}
The approximation ratio of the PCD mechanism is upper-bounded by \( 2 \). This bound is supported by the existence of a~sequence of profiles (for varying sets of agents) in which the approximation ratio of the PCD mechanism approaches \( 2 \)~\citep{meir2019strategyproof}.
The \( RD \) and \( PCD \) mechanisms are strategyproof, anonymous, neutral, and peaks-only.

\section{Analysis}

In this section, we prove the main result of the paper:
\begin{theorem}\label{thm:apx-1.75}
    For any set of agents with an odd cardinality and a cyclic graph \( G \) of length \( 1 \), there exists a strategyproof mechanism \( M \) whose approximation ratio is bounded from above by $7/4$.
\end{theorem}

\begin{remark}

    The above theorem can be easily generalized to a cycle \( G_z \) of arbitrary length \( z > 0 \). Let \( f \) be a mapping that uniformly scales the cycle \( G_z \) to \( G \). The mechanism \( M \) from \cref{thm:apx-1.75} can be adapted to \( G_z \) by combining it with the mapping \( f \) i.e., considering mechanism \( M' \) defined for every profile \( b\in G_z^N \) as \( M'(b)=M(f\circ b)\circ f \). This operation preserves both the strategyproofness and the approximation ratio of the mechanism.
\end{remark}

The proof of \cref{thm:apx-1.75} is constructive: we present a~mechanism that satisfies conditions of the theorem, i.e., it is SP and achieves an approximation ratio bounded by \( 7/4 \).

Let \( M_1 + M_2 \) denote a mechanism that is a \emph{mixture} of mechanisms \( M_1 \) and \( M_2 \).
The mechanism \( M_1 + M_2 \), for every profile \( b \), returns the outcome of \( M_1 \) on \( b \) with probability \( 1/2 \) and the outcome of \( M_2 \) on \( b \) with probability \( 1/2 \).
Formally, for any point \( v \in G \), the probability of selecting \( v \) in the resulting lottery \( l = (M_1 + M_2)(b) \) of the mixed mechanism is the average of its probabilities in the lotteries \( l_1 = M_1(b) \) and \( l_2 = M_2(b) \):
\begin{equation*}
    l(v) = \frac{l_1(v) + l_2(v)}{2}.
\end{equation*}

Mixing mechanisms has several desirable properties, making it a promising tool for constructing SP mechanisms
with low approximation ratio:
\begin{observation}\label{rem:wlasciwości-mieszania}
    For any mechanisms \( M_1 \) and \( M_2 \) that are strategyproof, anonymous, neutral, or peaks-only, their mixture \( M_1 + M_2 \) also satisfies these properties.
\end{observation}

Additionally, the approximation ratio of the mixed mechanism behaves as follows:
\begin{observation}\label{rem:apx-mieszania}
    For a profile \( b \), the approximation ratio of the mixed mechanism \( M_1 + M_2 \) is the arithmetic mean of the approximation ratios of \( M_1 \) and \( M_2 \):
    \begin{equation*}
        \begin{split}
            \apx_{M_1+M_2}(b) = \frac{\apx_{M_1}(b) + \apx_{M_2}(b)}{2}.
        \end{split}
    \end{equation*}
    This result follows directly from the definition of the approximation ratio and the linearity of the social cost.
\end{observation}

Consider the mechanism \( RD + PCD \).
By \cref{rem:wlasciwości-mieszania}, it is strategyproof.
Therefore, to conclude the proof of \cref{thm:apx-1.75}, it suffices to show that the approximation ratio of the \( RD + PCD \) mechanism is bounded above by \( 7/4 \).

Let us begin with some trivial bounds.

\begin{remark}\label{rem:trivial-bounds}
    The approximation ratio of the \( RD + PCD \) mechanism equals \( \left(\apx_{RD}(b) + \apx_{PCD}(b)\right)/2 \) (by \cref{rem:apx-mieszania}), which:
    \begin{itemize}
        \item cannot be bounded from above by any value smaller than \( 3/2 \), since there exist profiles for which approximation ratio of RD is arbitrarily close to \( 2 \), while approximation ratio of PCD is at least \( 1 \),
        \item is at most \( 2 \), under the conjecture that approximation ratio of \( PCD \) mechanism equals \( 2 \)~\citep{meir2019strategyproof}.
    \end{itemize}
\end{remark}

The bound of \( 7/4 \), stated in \cref{thm:apx-1.75}, lies in the middle between these two values and we will improve this result even further in \cref{sec:experiments}.

Approximation ratio of the \( RD + PCD \) mechanism, by definition, is the smallest \( \alpha \) such that inequality:
\begin{equation*}
    \soc_b((RD + PCD)(b)) \leq \alpha \cdot \opt_b
\end{equation*}
holds for every profile \( b \).
The above formula is not the best to work with and can be simplified if the optimal cost is greater than \( 0 \).
Let us start with considering the border case of \( \opt_b=0 \), to exclude it from further analysis:
\begin{remark}\label{rem:apx-rdpcd-1}
    Consider a profile \( b \), for which \( \opt_b=0 \).
    Then the approximation ratio of the \( RD + PCD \) mechanism is \( 1 \).
\end{remark}
\begin{proof}
    If \( \opt_b=0 \) then all the agents report the same point under \( b \).
    Let us denote this point by \( v \).
    By \cref{rem:wlasciwości-mieszania}, the \( RD+PCD \) mechanism is peaks-only and so its outcome on \( b \) is \( v \) with probability \( 1 \).
    Hence, the social cost of the mechanism output is \( 0 \) and the approximation ratio is \( 1 \).
\end{proof}

If \( \opt_b > 0 \), which we assume for the rest of the paper, the following holds:
\begin{equation*}
    \apx_{RD+PCD}(b)=\frac{\soc_b((RD + PCD)(b))}{\opt_b}.
\end{equation*}
Estimating the above expression for a mechanism operating on a cycle presents new challenges as compared to the case of a line segment:
\begin{enumerate}
    \item For the line segment, it is known that the optimal cost \( \opt_b \) corresponds to the social cost for the point preferred by the median agent. For the cycle, we are not aware of any compact description of \( \opt_b \).
    \item The social cost depends on the distances between the lottery points chosen by the mechanism and the agents' ideal points.
          Distances on a line segment have significantly simpler form in comparison to those on a cycle.
\end{enumerate}

We address the first of these problems in \cref{sec:normalizacja-stanow-swiata}, by selecting a subset of normalized profiles \( B' \subsetneq B \), for which the point generating optimal cost is fixed. In addition, we require that agents are ordered according to the clockwise traversal of the cycle, which simplifies further analysis.
We show that due to the anonymity and neutrality of the \( RD+PCD \) mechanism,
set \( B' \) is representative in terms of the approximation ratio values achieved by \( RD+PCD \),
i.e., for every \( b \in B \), there exists \( b' \in B' \) such that \( \apx_{RD+PCD}(b) = \apx_{RD+PCD}(b') \).
This allows us to narrow down further considerations to the set of normalized profiles \( B' \).

The second problem, of the more complex nature of distances between points on a cycle, is addressed in \cref{sec:cykl-cutting}.
We estimate these distances through their counterparts on a~line segment obtained by cutting the cycle before some point.
Such estimation increases the social costs associated with certain points of the graph, which could potentially increase the optimal cost
(and thus decrease the estimated value yielding invalid bound).
Fortunately, due to the earlier restriction to the set of normalized profiles \( B' \), the point minimizing the social cost is fixed.
This enables us to select the cutting point so that the minimal social cost is preserved, thereby obtaining a valid estimate.

In \cref{sec:auxiliary-estimations}, we provide concise formulas for parts of the approximation ratio after the cut of the \( RD+PCD \) mechanism corresponding to the social costs of the \( RD \) and \( PCD \) mechanisms.
We substitute these expressions into the formula for the approximation ratio after the cut of the \( RD+PCD \) mechanism, obtaining an estimate \( \phi \).
Finally, in \cref{sec:do-brzegu}, we upper-bound the value of \( \phi \) by \( 7/4 \) on a set of normalized profiles via a series of technical lemmas.

The formal proof of \cref{thm:apx-1.75}, based on the developed lemmas, is presented in \cref{sec:proof-main-theorem}.

\subsection{Normalization of profiles}\label{sec:normalizacja-stanow-swiata}
We start with the definition of a subset of profiles for which the formula for the approximation ratio of the \( RD+PCD \) mechanism is simplified, but still can obtain the same values.
\begin{definition}
    A profile \( b \) is \emph{normalized} if:
    \begin{itemize}
        \item point \( 0 \) minimizes the social cost in \( b \) over all the points of the cycle,
        \item \( b \) is a non-decreasing function with respect to the natural order of agents and the clockwise order of the points of the cycle,
        \item agent \( 0 \) reports the point \( 0 \).
    \end{itemize}
\end{definition}

The following lemma allows us to restrict  attention to the normalized profiles.
\begin{restatable}{lemma}{lemZnormalizowaneStanySwiata}\label{lem:normalizacja-stanow-swiata}
    Let \( M \) be any mechanism that is anonymous and neutral.
    Let \( b \) be any profile.
    There exists a normalized profile, \( b' \), such that \( \apx_M(b) = \apx_M(b') \).
\end{restatable}

\begin{proof}
    Let \( b \) be any profile.
    A normalized profile \( b' \) with the same approximation ratio is constructed as follows:
    \begin{enumerate}
        \item By the neutrality of the \( M \) mechanism,
              we cyclically shift the naming of the points of the graph, so that
              point \( 0 \) generates the minimal social cost.
        \item By the anonymity property, we can freely rename agents without affecting the result of the mechanism.
              We traverse the cycle \( G \) clockwise starting from the joined endpoint \( -0.5 \),
              numbering agents in the order their reports are encountered (breaking ties arbitrarily).
    \end{enumerate}

    Due to the properties of the point generating the minimal social cost on the cycle,
    the above two steps ensure that the third condition of normalization is satisfied,
    namely, that agent \( 0 \) reports point \( 0 \).
    Assume, to the contrary, that this is not true, i.e., \( b'_0 \neq 0 \).
    Firstly consider the case \( b'_0 > 0 \).
    By step \( 2 \), \( b' \) is non-decreasing, so any agent \( i > 0 \) reports a~point
    greater than or equal to \( b'_0 \), i.e., \( b'_i \geq b'_0 \).
    At the same time, \( b'_i \leq 0.5 \) since agents selecting negative points are numbered before agent \( 0 \).
    It follows that moving from point \( 0 \) to point \( b'_0 \)
    brings us closer to the reports of strictly more than half of the agents (all agents with indices \( i \geq 0 \)).
    Thus, the social cost of point \( b'_0 \) is less than the social cost of point \( 0 \),
    a contradiction to the assumption that point \( 0 \) generates the minimal social cost.
    The case of \( b'_0 < 0 \) is handled analogously.
\end{proof}

\subsection{Cut of the profile}\label{sec:cykl-cutting}

After introducing the concept of normalization of profiles, let us proceed to analyzing the approximation ratio of the \( RD+PCD \) mechanism.
This quantity, by definition, depends on the social cost,
which, in turn, depends on the distances between the agents' ideal points and the points in the support of the output lottery of the mechanism.
The definition of distance between the points on a cycle involves a~minimum.
This non-linearity makes a direct analysis of the approximation ratio of the \( RD+PCD \) mechanism challenging.
We address this issue by replacing the distance function used to compute the social cost with a simpler one.
Let \( d' \) denote a line-segment-like distance function on cycle \( G \), defined for any two points \( v_1, v_2 \in G \) as \( d'(v_1, v_2) = |v_2 - v_1| \).
Let us denote quantities such as cost, social cost, and the approximation ratio computed with respect the distance function $d'$, as \( c', \soc',\apx' \), respectively.

For any pair of points \( v_1, v_2 \in G \), the metric \( d \) is defined as \( d(v_1, v_2) = \min(|v_2 - v_1|, 1 - |v_2 - v_1|) \).
This corresponds to the general definition of distance on graphs as the length of the shortest path between two points.
The metric \( d' \) is derived from \( d \) by removing one of the two components of the minimum.
Conceptually, this is equivalent to excluding paths that traverse the joined endpoint of the cycle when defining the distance.
Intuitively, we can imagine that for the purpose of measuring distances, the cycle is \quo{cut} just before point \( -0.5 \).
Due to this intuition, the transition from the cycle with the metric \( d \) to the cycle with the metric \( d' \) will henceforth be referred to as the \emph{cut} of the cycle.

The following inequality connects the value of the social cost for different points of the cycle \( G \) computed before and after the cut (with respect to the original distance function, \( d \), and the new distance function, \( d' \)):

\begin{lemma}\label{lem:rozciecie}
    Let \( b \) be any normalized profile.
    For any lottery \( l \in \Delta(G) \), it holds that:
    \begin{equation*}
        \soc'_b(l) \ge \soc_b(l).
    \end{equation*}
    Moreover, the social cost associated with the point \( 0 \) is preserved by the cut,
    i.e., \( \soc'_b(0) = \soc_b(0) \).
\end{lemma}

\begin{proof}
    The social cost of a point \( v \in G \) is defined as a linear combination of distances between \( v \) and the agents' ideal points.
    This definition is extended to lotteries by taking expected values.
    The first statement holds because \( d' \) never decreases the distances between the points, as compared to \( d \),
    and therefore the values on which the social cost is based can only increase after the cut.
    The second statement follows from the fact that the cut does not alter the shortest path between \( 0 \) and any other point.
\end{proof}

Based on \cref{lem:rozciecie} and \cref{rem:apx-mieszania}, we can establish the following upper bound on the approximation ratio of the \( RD+PCD \) mechanism:
\begin{corollary}\label{cor:szacowanie-apx}
    For any normalized profile \( b \), the approximation ratio of the \( RD+PCD \) mechanism is bounded by
    \begin{equation*}
        \begin{split}
            \phi(b):=\frac{\soc'_b(RD(b)) + \soc'_b(PCD(b))}{2\soc'_b(0)}.
        \end{split}
    \end{equation*}
\end{corollary}
\begin{proof}
    \begin{equation*}
        \begin{split}
            \apx_{RD+PCD}(b) = \frac{\apx_{RD}(b) + \apx_{PCD}(b)}{2} =
            \frac{\soc_b(RD(b)) + \soc_b(PCD(b))}{2\opt_b} = \\
            \frac{\soc_b(RD(b)) + \soc_b(PCD(b))}{2\soc_b(0)} \le
            \frac{\soc'_b(RD(b)) + \soc'_b(PCD(b))}{2\soc'_b(0)}=\phi(b).\qedhere
        \end{split}
    \end{equation*}
\end{proof}

\subsection{Estimations of the social cost after the cut}\label{sec:auxiliary-estimations}

In the previous subsection we bounded the value of the approximation ratio of the \( RD+PCD \) mechanism by the expression \( \phi(b) \),
which depends on \( \soc'_{RD}(b), \soc'_{PCD}(b) \): the social costs of outputs of \( RD \) and \( PCD \) after the cut.
In this subsection we provide concise forms for these quantities.
Since deriving those values involves only algebraic manipulations, we postpone the detailed proofs to the appendix.

\begin{restatable}{lemma}{lemkosztRDnasciezce}\label{lem:koszt-RD-na-sciezce}
    For any normalized profile \( b \), the social cost of the outcome of mechanism \( RD \) after the cut is equal to:
    \begin{equation*}
        \soc'_{b}(RD(b)) = \sum_i \frac{4|i|}{2k+1}|b_i|.
    \end{equation*}
\end{restatable}

\begin{restatable}{lemma}{lemkosztPCDnaodcinku}\label{lem:koszt-PCD-na-odcinku}
    For any normalized profile \( b \) the social cost of the outcome of mechanism \( PCD \) after the cut is equal to:
    \begin{equation*}
        \begin{split}
            \soc'_{b}(PCD(b)) =
            \sum_{j > 0} |b_j| (2k+1 - 2j)(|b_{j-k-1}| - |b_{j-k}|) +\\
            \sum_j|b_j|
            + \sum_{j < 0} |b_j| (2k+1 + 2j)(|b_{j+k+1}| - |b_{j+k}|).
        \end{split}
    \end{equation*}
\end{restatable}

\subsection{Estimating $\phi(b)$}\label{sec:do-brzegu}

Substituting the results from \cref{sec:auxiliary-estimations} into the expression for \( \phi(b) \), allows us to express it as
a fraction with numerator:
\begin{equation}\label{eq:phi-numerator}
    \begin{split}
        \sum_i \frac{4|i|}{2k+1}|b_i| +
        \sum_j|b_j| +
        \sum_{j > 0} |b_j| (2k+1 - 2j)(|b_{j-k-1}| - |b_{j-k}|) + \\
        \sum_{j < 0} |b_j| (2k+1 + 2j)(|b_{j+k+1}| - |b_{j+k}|)
    \end{split}
\end{equation}
and denominator:
\begin{equation}\label{eq:phi-denominator}
    2\soc'_b(0) = 2\sum_i d(b_i,b_0)=2\sum_i |b_i-0|=2\sum_i |b_i|.
\end{equation}
In this section, we bound the value of \( \phi(b) \) by \( 7/4 \) for any normalized profile \( b \). To achieve this, we break the proof into several steps, stated as lemmas. We provide some intuition behind each lemma and defer the detailed proofs to the appendix.

We begin by extending the set of considered profiles, as the current set (the normalized profiles) is not easily described due to the requirement that point \( 0 \) minimizes the social cost. To address this, we define a larger set, \( B'' \), which contains all the non-decreasing functions from the set of agents to \( \cc{-0.5}{0.5} \), except the constant \( \bm{0} \) function, such that \( b''_0 = 0 \). One can view the elements of \( B'' \) as profiles on a closed line segment of length \( 1 \), where agent \( 0 \) reports the point \( 0 \).

We extend the definition of \( \phi \) to be valid on \( B'' \) and proceed to estimating the supremum of \( \phi \) over \( B'' \). This approach provides an upper bound for the supremum over the normalized profiles, as \( B'' \) is a larger set.

A challenging aspect of estimating the supremum of \( \phi \) over \( B'' \) is that \( B'' \) is not a closed set, because it excludes \( \bm{0} \). To address this, we begin by proving that \( \phi \) attains small values in the neighborhood of \( \bm{0} \).

We introduce the concept of a \emph{dominated} function in \( B'' \): a function \( b \in B'' \) is dominated if it equals zero for at least \( k+1 \) agents. Let \( B_0'' \subseteq B'' \) denote the subset of functions in \( B'' \) that are dominated.

\begin{lemma}\label{lem:szacowanie-zdominowanych}
    For any dominated \( b \in B_0'' \), it holds that \( \phi(b) \le 3/2-1/n \).
    Moreover, the equality can be achieved as witnessed by the function \( b^@\in B_0'' \) corresponds to tuple \( (0,0,\ldots,0.5) \)
    i.e., the function getting value \( 0 \) for every agent except agent \( k \) for which it gets value \( 0.5 \).
\end{lemma}


Let us define function \( w:B'' \to \mathbb{N} \) which, for any \( b \in B'' \), returns the number of distinct values of \( b \) that are not equal to \( -0.5 \), \( 0 \), or \( 0.5 \). Formally,
\[
    w(b) = \left| \Ima b \setminus \{-0.5, 0, 0.5\} \right|.
\]
We call a function \( b \in B'' \) \emph{boundary} if \( w(b) = 0 \). Let \( B^* \) denote the set of all boundary functions in \( B'' \).

We demonstrate that the supremum of \( \phi \) over \( B'' \) is the same as its supremum over the smaller set \( B^* \). To establish this, we start with a lemma that allows us to reduce the number of non-boundary values without decreasing the value of \( \phi \).

\begin{restatable}{lemma}{lemSupBStar}\label{lem:supremum-na-B*}
    Let \( b \in B'' \) be a function with non-boundary values, i.e., \( w(b) > 0 \). There exists \( b' \in B'' \) such that
    \( w(b') < w(b) \) and \( \phi(b') \geq \phi(b) \) holds.
\end{restatable}

The proof proceeds as follows: we consider a pair of functions \( b_1, b_2 \in B'' \), constructed by replacing one non-boundary value in \( b \) (for all agents reporting it) with either the value reported by other agents or with a boundary value. By construction, both \( b_1 \) and \( b_2 \) have smaller \( w \)-values than \( b \). To conclude the proof, we show that \( \phi \) for at least one of these modified functions is not smaller than \( \phi(b) \).

By repeatedly applying the above lemma, we can iteratively reduce the number of non-boundary values in any \( b \in B'' \) until \( w(b) = 0 \), while ensuring that \( \phi \) does not decrease. Thus, we have:
\begin{corollary}\label{cor:supremum-na-B*}
    For any element \( b \in B'' \), there exists a boundary element \( b^* \in B^* \) such that \( \phi(b) \leq \phi(b^*) \).
\end{corollary}

From \cref{cor:supremum-na-B*}, we immediately deduce:
\begin{equation*}
    \sup_{b \in B''} \phi(b) \leq \sup_{b \in B^*} \phi(b).
\end{equation*}
Finally, we bound the supremum of \( \phi \) over \( B^* \) by \( 7/4 \).

\begin{restatable}{lemma}{lemApxNaBStar}\label{lem:apx-na-B*}
    For any \( b^* \in B^* \), the value of \( \phi \) for \( b^* \) is at most \( 7/4 \).
\end{restatable}

We achieve this by simplifying the representation of boundary functions in \( B^* \) and directly estimating the value of \( \phi \) for these functions.

\subsection{Proof of the main theorem}\label{sec:proof-main-theorem}

We are now ready to prove \cref{thm:apx-1.75}.

\begin{proof}[Proof (of \cref{thm:apx-1.75})]
    Consider the mechanism \( RD+PCD \).
    By \cref{rem:wlasciwości-mieszania}, this mechanism is strategyproof.
    Now, let us show that the approximation ratio of \( RD+PCD \) mechanism on all profiles on cycle \( G \) is bounded by \( 7/4 \).
    Let \( b \) be any profile on \( G \).
    If all agents report the same point, then the approximation ratio of the \( RD+PCD \) mechanism is equal to \( 1 \), as shown in \cref{rem:apx-rdpcd-1}.
    Assume now that there are at least two distinct reports in the profile \( b \).
    By \cref{rem:wlasciwości-mieszania}, the \( RD+PCD \) mechanism is anonymous and neutral.
    Hence, by \cref{lem:normalizacja-stanow-swiata}, there exists a normalized profile, \( \tilde{b} \), on the cycle such that $\apx_{RD+PCD}(b) = \apx_{RD+PCD}(\tilde{b})$.
    By \cref{cor:szacowanie-apx} we have that $\apx_{RD+PCD}(\tilde{b}) \le \phi(\tilde{b})$.
    By \cref{cor:supremum-na-B*}, there exists a boundary function \( b^*\in B^* \)
    such that \( \phi(\tilde{b})\le\phi(b^*) \).
    Lastly, by \cref{lem:apx-na-B*}, we have that $\phi(b^*) \le 7/4$.
    Combining the above inequalities we obtain the desired bound for the approximation ratio of the \( RD+PCD \) mechanism for the profile \( b \).
    Since the profile \( b \) was arbitrary, this completes the proof.
\end{proof}

\section{Experiments}\label{sec:experiments}

In the previous section, we proved an upper bound of \( 7/4 \) on the approximation ratio of the \( RD + PCD \) mechanism.
On the other hand, in \cref{rem:trivial-bounds}, we established that no upper bound below \( 3/2 \) is possible, leaving some uncertainty about the actual approximation ratio of the \( RD + PCD \) mechanism.
To complement these theoretical results, we conducted numerical experiments to investigate which values of the approximation ratio are actually attained by the \( RD + PCD \) mechanism.

\subsection{Setup}
The idea of the experiments is to compute the approximation ratio of the \( RD+PCD \) mechanism for a finite, computationally feasible, set of profiles, providing insight into the behavior of the mechanism for all possible profiles.

The first obvious parameter to consider is the number of agents \( n \), as it directly affects the definition, and potentially the operation of, the \( RD+PCD \) mechanism. During the experiments, we considered the number of agents \( n \) in the range \( [2..60] \).

The second parameter to consider is the set of possible reports for each agent.
In the theoretical analysis, no restrictions were imposed on the agents' reports.
However, since cycle \( G \) is a continuous set with infinitely many points, we introduce such restrictions in the experiments in order to have a finite set of profiles to analyze.
For \( l \in \mathbb{N} \), let \( G_l \) denote the subset of \( l \) points on the cycle \( G \) (of length \( 1 \)) that are equally spaced and include the point \( 0 \).
During the experiments, we restricted agents' reports to the set \( G_l \) for \( l \in [2..60] \).
We hypothesize that profiles with agents' reports restricted to the set \( G_l \) provide increasingly better approximation of the behavior of the \( RD+PCD \) mechanism for all possible profiles, as \( l \) grows (since the allowed reports cover the cycle more densely).

The total number of distinct profiles with \( n \) agents whose reports are restricted to \( G_l \) is \( l^n \). To reduce the computational complexity, we leveraged the anonymity and neutrality of the \( RD+PCD \) mechanism, grouping profiles that would yield the same social cost and avoiding redundant calculations. Despite this, for larger values of \( n \) and \( l \), the number of profiles still exceeded \( 10^7 \), making computations infeasible.

In such cases, we restricted our analysis to profiles in which agents reported no more than \( 3 \) distinct points. These profiles were considered good candidates for generating the highest approximation ratio for the \( RD+PCD \) mechanism, as will be demonstrated later.

\subsection{Results}
\begin{figure}
    \centering
    \includegraphics[width=0.7\columnwidth]{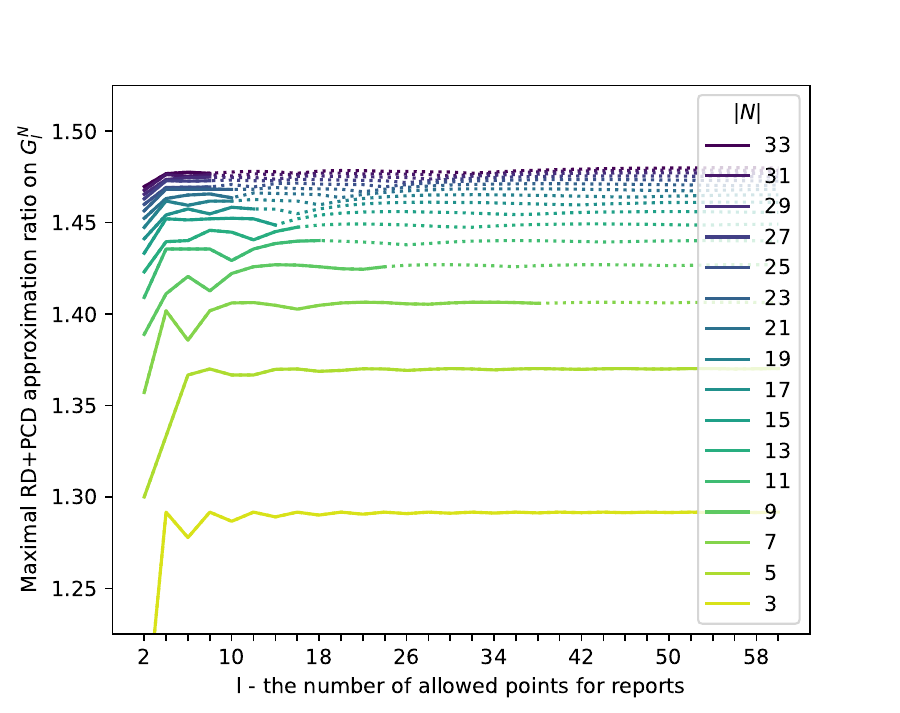}
    \caption{
        The maximum value of the approximation ratio of the \( RD+PCD \) mechanism in profiles with \( n \) agents, whose reports are restricted to \( G_k \).
        The dotted line highlights the results calculated with the restriction that the number of different points reported by agents is limited to \( 3 \).
    }\label{fig:rdpcd-num-of-points}
\end{figure}

The results of the experiments are presented in \cref{fig:rdpcd-num-of-points}. The graph illustrates the maximum value of the approximation ratio of the \( RD+PCD \) mechanism for profiles with \( n \) agents, whose reports are restricted to \( G_l \). The dotted line represents results computed under the restriction that agents report no more than \( 3 \) distinct points.

An immediate observation is that the approximation ratio of the \( RD+PCD \) mechanism does not exceed \( 3/2 \) for any profile considered in the experiments. This observation leads to the following hypothesis, which, if true, would further refine the theoretical result established in \cref{thm:apx-1.75}:
\begin{hypothesis}\label{hyp:rdpcd-1.5}
    The approximation ratio of the \( RD+PCD \) mechanism is strictly less than \( 3/2 \) for any set of agents and any profile on the cycle \( G \).
\end{hypothesis}

The bound of \( 3/2 \) is the lowest possible due to \cref{rem:apx-mieszania} and would imply that the mechanisms \( RD \) and \( PCD \) perfectly complement each other, meaning the worst-case profile for one mechanism is the best-case profile for the other.

A deeper analysis of the experimental results provides further evidence supporting \cref{hyp:rdpcd-1.5}. Indeed, we observe that for a fixed number of agents, increasing the number of points on which agents can report (i.e., increasing \( i \)) beyond a certain threshold does not affect the highest approximation ratio of the \( RD+PCD \) mechanism (excluding parity effect). This suggests that the approximation ratio of the \( RD+PCD \) mechanism does not grow significantly if we were to extend the range of \( l \).

\section{Conclusions}

We considered a strategyproof facility location on a cycle with the utilitarian welfare objectives. We showed that the approximation ratio of strategyproof mechanisms for this problem is bounded from above by $7/4$ and is guaranteed by the $RD+PCD$ mechanism. This bound improves the best previously known bound of $2 - 2/n$ in the cases of odd $n \geq 5$. Computational analysis of the approximation ratio of the $RD+PCD$ mechanism suggests that the bound could be further improved to $3/2$. Verifying this hypothesis is an~interesting question for future research.

\appendix
\section*{Acknowledgements}
This work was supported by the Polish National Science Centre through
grant 2021/42/E/HS4/00196.

\bibliographystyle{unsrtnat}


\iftrue

    \clearpage

    \section{Appendix}
    Respective sections of the appendix contain detailed proofs of the lemmas presented in the main text.
    \subsection{Estimations of social cost after the cut}

    Before proving \cref{lem:koszt-RD-na-sciezce} and \cref{lem:koszt-PCD-na-odcinku}, we derive some auxiliary results about the value of social cost after the cut:

    \begin{lemma}\label{różnica-sąsiednich-kosztów}
        Fix an arbitrary normalized profile \( b \) on the cycle.
        Let \( i \in N \) be any agent with \( i>0 \).
        After the cut, the following relation holds between the social costs
        associated with the reports of agents \( i, i-1 \) and \( 0 \):
        \begin{enumerate}
            \item \(\soc'_{b}(b_i) - \soc'_{b}(b_{i-1}) = (2i-1)(|b_i| - |b_{i-1}|)\).
            \item \(\soc'_{b}(b_i) - \soc'_{b}(b_{0}) = (2i-1)|b_i| - \sum_{j=1}^{i-1} 2|b_j|\)
        \end{enumerate}
    \end{lemma}

    \begin{proof}
        We start with proving the first part of the lemma.
        By the definition of the social cost after the cut, we have:
        \begin{equation*}
            \begin{split}
                \soc'_{b}(b_i) - \soc'_{b}(b_{i-1}) = \sum_{j} (d'(b_j, b_i) - d'(b_j, b_{i-1})) \overset{(1)}{=} \\
                \sum_{j \le i-1} (|b_i| - |b_{i-1}|) - \sum_{j \ge i} (|b_i| - |b_{i-1}|) \overset{(2)}{=}
                (2i-1)(|b_i| - |b_{i-1}|),
            \end{split}
        \end{equation*}
        where the equalities follow from:
        \begin{enumerate}
            \item expanding the difference \( d(b_j, b_i) - d(b_j, b_{i-1}) \) depending on whether \( j \le i-1 \),
            \item counting the number of elements in each sum.
        \end{enumerate}

        The second part of the lemma follows from the first part.
        By summing the incremental differences, we have:
        \begin{equation*}
            \begin{split}
                \soc'_{b}(b_i) - \soc'_{b}(b_{0}) = \sum_{j=1}^i (\soc'_{b}(b_j) - \soc'_{b}(b_{j-1})) \overset{(1)}{=}
                \sum_{j=1}^i (2j-1)(|b_j| - |b_{j-1}|) \overset{(2)}{=}
                (2i-1)|b_i| - \sum_{j=1}^{i-1} 2|b_j|,
            \end{split}
        \end{equation*}
        where the equalities follow from:
        \begin{enumerate}
            \item repeated application of already proven first part of the lemma,
            \item grouping terms by \( |b_{j}| \) and noting that \( |b_{0}| = 0 \).
        \end{enumerate}
    \end{proof}

    \begin{remark}\label{rem:roznica-sasiadujacych-kosztow2}
        The social cost after the cut has analogous formula for agents with indices of the same absolute value.
        Therefore, for any agent \( i<0 \) following the proof of \cref{różnica-sąsiednich-kosztów} we can establish that following holds:
        \[\soc'_{b}(b_i) - \soc'_{b}(b_{i+1}) = (-2i-1)(|b_i| - |b_{i+1}|).\]
    \end{remark}

    Using \cref{różnica-sąsiednich-kosztów} and \cref{rem:roznica-sasiadujacych-kosztow2}, we derive expressions for the social costs, in terms of the agents' placement, after the cut of outcomes of the mechanisms \( RD \) and \( PCD \), as stated in the main text.

    \begin{proof}[Proof (of \cref{lem:koszt-RD-na-sciezce})]
        Let us expand the definition of the social cost after the cut for the \( RD \) mechanism:
        \begin{equation*}
            \begin{split}
                \soc'_{b}(RD(b)) = \frac{\sum_i \soc'_{b}(b_{i})}{2k+1} =
                \frac{\sum_i (\soc'_{b}(b_{i}) - \soc'_{b}(b_{0}))}{2k+1} + \soc'_{b}(b_{0}).
            \end{split}
        \end{equation*}
        First, consider the sum in the numerator of the fraction in the above expression.
        Focus on the part for positive indices. By applying second part of \cref{różnica-sąsiednich-kosztów}, we have:
        \begin{equation*}
            \begin{split}
                \sum_{i>0}(\soc'_{b}(b_{i}) - \soc'_{b}(b_{0})) = \sum_{i>0}(2i - 1)|b_{i}| - \sum_{i>0} \sum_{j<i} 2|b_{j}| =
                \sum_{i>0}(4i - 2k - 1)|b_{i}|,
            \end{split}
        \end{equation*}
        where the last step follows from reordering the summation:
        \begin{equation*}
            \begin{split}
                \sum_{i>0} \sum_{j<i} 2|b_{j}| = \sum_{j>0} \sum_{i>j} 2|b_{j}| =
                \sum_{i>0} 2(k - i)|b_{i}|.
            \end{split}
        \end{equation*}
        Thus, the coefficient for \( |b_{i}| \) is \( 4|i| - 2k - 1 \), for \( i > 0 \).

        For negative values of \( i \), due to symmetry of social costs after the cut, the same relationship holds.
        It is also valid for \( i = 0 \), because in that case:
        \begin{equation*}
            \soc'_{b}(b_{0}) - \soc'_{b}(b_{0}) = 0 = (4i - 2k - 1) \cdot 0 = (4i - 2k - 1)|b_{0}|.
        \end{equation*}
        Hence, we have:
        \begin{equation*}
            \sum_i (\soc'_{b}(b_{i}) - \soc'_{b}(b_{0})) = \sum_i (4|i| - 2k - 1)|b_{i}|.
        \end{equation*}
        Substituting this result into the definition of the social cost after the cut for the \( RD \) mechanism gives:
        \begin{equation*}
            \begin{split}
                \soc'_b(RD(b)) = \frac{\sum_i (\soc'_{b}(b_{i}) - \soc'_{b}(b_{0}))}{2k+1} + \soc'_{b}(b_{0}) =
                \frac{\sum_i (4|i| - 2k - 1)|b_{i}| + (2k+1)\sum_i |b_{i}|}{2k+1} =
                \frac{\sum_i 4|i||b_{i}|}{2k+1}.\qedhere
            \end{split}
        \end{equation*}
    \end{proof}

    \begin{proof}[Proof (of \cref{lem:koszt-PCD-na-odcinku})]
        By directly expanding the definition of the \( PCD \) mechanism, we have:
        \begin{equation*}
            \begin{split}
                \soc'_b(PCD(b)) =
                \sum_{i < 0} \soc'_{b}(b_{i}) (|b_{i+k+1}| - |b_{i+k}|)
                + \soc'_{b}(b_{0}) (1 - |b_{k}| - |b_{-k}|)
                + \sum_{i > 0} \soc'_{b}(b_{i}) (|b_{i-k-1}| - |b_{i-k}|).
            \end{split}
        \end{equation*}
        Let us analyze the first summation in the above expression. We obtain:
        \begin{equation*}
            \begin{split}
                \sum_{i < 0} \soc'_{b}(b_{i}) (|b_{i+k+1}| - |b_{i+k}|) \overset{(1)}{=}
                \sum_{j > 0} \soc'_{b}(b_{j-k-1}) (|b_{j}| - |b_{j-1}|) \overset{(2)}{=}                       \\
                |b_{k}| \soc'_{b}(b_{0}) - |b_{0}| \soc'_{b}(b_{-k})  +
                \sum_{j > 0} |b_{j}| (\soc'_{b}(b_{j-k-1}) - \soc'_{b}(b_{j-k})) \overset{(3)}{=} \\
                |b_{k}| \soc'_{b}(b_{0}) + \sum_{j > 0} |b_{j}| (2k - 2j + 1)(|b_{j-k-1}| - |b_{j-k}|),
            \end{split}
        \end{equation*}
        where the steps follow from:
        \begin{enumerate}
            \item substituting the indices \( i = j - k - 1 \),
            \item grouping terms by \( |b_{j}| \),
            \item applying \cref{rem:roznica-sasiadujacych-kosztow2} and noting \( |b_{0}| = 0 \) (which holds because \( b \) is normalized).
        \end{enumerate}
        Applying a similar transformation to the last summation and substituting the obtained equalities into the formula for the social cost after the cut of \( PCD \) completes the proof.
    \end{proof}

    \subsection{Estimating $\phi(b)$}
    Recall that the expression for \( \phi \) was expanded at the beginning of \cref{sec:do-brzegu} as a fraction, with the numerator and denominator given by \cref{eq:phi-numerator} and \cref{eq:phi-denominator}, respectively.

    \begin{proof}[Proof (of \cref{lem:szacowanie-zdominowanych})]
        Consider the numerator of \( \phi \) as expanded in \cref{eq:phi-numerator}.
        Statement of the lemma is equivalent to showing that it is bounded by \( (3/2 - 1/n) \) times the value of the denominator which, by \cref{eq:phi-denominator}, is \( D = 2\sum_i |b_i| \).

        We start with the part of the numerator corresponding to the \( RD \) mechanism:
        \begin{equation*}
            \sum_i \frac{4|i|}{2k+1}|b_i| \le
            \frac{4k}{2k+1} \sum_i |b_i| =
            \left( 1 - \frac{1}{n} \right) D,
        \end{equation*}
        where the last equality follows from \( 2k+1 = n \).

        The part of the numerator corresponding to \( PCD \) is a sum of three terms.
        For the first of them, we trivially have \( \sum_j |b_j| = D/2 \).
        Since the above two parts of the numerator already sum up to the expected bound \( (3/2 - 1/n)D \),
        to finish the proof we need to show that the remaining two parts are zeros.

        Consider the first of these remaining terms:
        \begin{equation}\label{val:pcd}
            \sum_{j > 0} |b_j| (2k+1 - 2j)(|b_{j-k-1}| - |b_{j-k}|).
        \end{equation}
        We will show that every term in the above sum is zero.

        Let \( i \) be the first (largest) agent for which \( b_i = 0 \) (such agent exists because \( b_0=0 \)).
        By the assumption that \( b \) is dominated,
        we know that it attains zero for at least \( k+1 \) agents.
        Moreover, it is non-decreasing, so we have \( b_m = 0 \) for every \( i-k \leq m \leq i \).
        The largest possible value of agent \( i \) is \( k \), hence \( i-k \leq 0 \).

        Therefore, the only terms of the sum in \cref{val:pcd} for which the first element \( |b_j| \) of the product is non-zero are those corresponding to \( j > i \).
        But for such \( j \), the last element of the product \( |b_{j-k-1}| - |b_{j-k}| = 0 \),
        because \( i-k \leq j-k-1 < j-k \leq 0 \leq i \), and hence \( |b_{j-k-1}| = 0 \) and \( |b_{j-k}| = 0 \).
        Thus, all terms of the sum in \cref{val:pcd} are zeros.

        A similar argument shows that every term in the sum of the last part of the numerator is also zero.

        This completes the proof of the first part of the lemma.

        Regarding second part of the lemma, the above estimations for the \( PCD \) part of the numerator are tight, independently of the considered function \( b \).
        Therefore, to conclude the proof it is enough to show that the estimation of the \( RD \) part of the numerator is tight for function \( b^@ \).
        For this function we have:
        \begin{equation*}
            \begin{split}
                D=2\sum_i |b_i| = 1
            \end{split}
        \end{equation*}
        and hence:
        \begin{equation*}
            \begin{split}
                \sum_i \frac{4|i|}{2k+1}|b_i| = \frac{4k}{2k+1}0.5 = \left( 1-1/n \right)D.\qedhere
            \end{split}
        \end{equation*}
    \end{proof}

    Before proving \cref{lem:supremum-na-B*}, we show an auxiliary lemma, that will be used in the proof.

    \begin{lemma}\label{lem:monotonicznosc-funkcji-homograficznej}
        Let \( f:\mathbb{R} \to \mathbb{R} \) be a function that can be expressed as \( f(x) = g(x)/h(x) \),
        where \( g,h:\mathbb{R} \to \mathbb{R} \) are linear functions of \( x \).
        Let \( A \subsetneq \mathbb{R} \) be any interval such that \( h(x) \neq 0 \) for every \( x \in A \).
        Then the function \( f \) is monotonic on \( A \).
    \end{lemma}

    \begin{proof}
        Intuitively, under the assumptions of the lemma, \( f \) is a homographic function.
        From well-known properties, a~homographic function is monotonic on any interval that does not contain a zero of its denominator.

        Formally, since \( g \) and \( h \) are linear in \( x \), there exist constants \( a,b,c,d \in \mathbb{R} \) such that:
        \begin{equation*}
            \begin{split}
                g(x) = ax + b, \quad
                h(x) = cx + d.
            \end{split}
        \end{equation*}
        The function \( f \) is differentiable on \( A \) (as \( h(x) \neq 0 \) on \( A \)), and its derivative is:
        \begin{equation*}
            \begin{split}
                f'(x) = \frac{g'(x)h(x) - g(x)h'(x)}{h(x)^2} =
                \frac{a(cx + d) - c(ax + b)}{(cx + d)^2} =
                \frac{ad - bc}{(cx + d)^2}.
            \end{split}
        \end{equation*}
        The numerator \( ad - bc \) is constant, while the denominator \( (cx + d)^2 \) is the square of a non-zero linear function, and hence always positive.
        Thus, the sign of the derivative is constant, which implies that \( f \) is monotonic on \( A \).
    \end{proof}

    \begin{proof}[Proof (of \cref{lem:supremum-na-B*})]
        If \( b \) is dominated then, by \cref{lem:szacowanie-zdominowanych}, it holds that \( \phi(b) \leq 3/2 - 1/n = \phi(b^@) \)
        (where \( b^@ \) comes from the statement of \cref{lem:szacowanie-zdominowanych}).
        It is trivial to check that \( b^@ \in B^* \) and fulfills the conditions of the lemma.
        This completes the proof for dominated \( b \).

        If \( b \) is not dominated then it attains at least three distinct values,
        since \( 0 \) is attained at most \( k \) times by \( b \).
        Function \( b \) is non-decreasing and \( b_0 = 0 \), hence \( b_{-k} < b_0 < b_k \).
        Moreover, by assumption, \( b \) is non-boundary, so there exists an~element \( v \) in \( \Ima b \) other than \( -0.5, 0, 0.5 \).
        Without loss of generality, assume that \( v \) is positive, i.e., \( v \in \oo{0}{0.5} \).

        Let \( v_{-1} \) denote the largest element in the image of \( b \) less than \( v \)
        (such an element always exists because \( 0 = b_0 < v \)).
        Let \( v_1 \) denote the smallest element in the image of \( b \) greater than \( v \), or \( 1/2 \) if no such element exists.
        Let \( A = b^{-1}(v) \) be the set of all agents for which the value of \( b \) equals \( v \).
        For any \( x \in \mathbb{R} \), let \( b[x] \) denote the function obtained from \( b \)
        by setting the values corresponding to all agents in \( A \) to \( x \), i.e., \( b[x] = b[a \to x : a \in A] \).

        Consider the functions \( b[v_1] \) and \( b[v_{-1}] \).
        By construction, \( \Ima b[v_1] = \Ima b[v_{-1}] = (\Ima b) \setminus \set{v} \), and hence
        \( w(b[v_1]) = w(b[v_{-1}]) < w(b) \), because \( v \), by construction, is non-boundary.

        Assume, for the sake of contradiction, that neither \( b[v_1] \) nor \( b[v_{-1}] \) satisfies the claim of the lemma.
        Since both those functions have fewer non-boundary values, this implies that \( \phi \)
        must be strictly smaller for them than for \( b \).
        We will now show that this is impossible.

        Let \( V = \cc{v_{-1}}{v_1} \).
        Define the function \( f_{b}: V \to \mathbb{R} \) as:
        \[
            f_{b}(x) = \phi(b[x]).
        \]
        We verify that \( f_{b} \) is well-defined, i.e., for any \( x \in V \), function \( b[x] \) belongs to the domain of \( \phi \), i.e., \( B'' \).
        Let \( x \in V \) be arbitrary.
        The transition from \( b \) to \( b[x] \) does not alter the relative order of agent values, maintaining monotonicity of the function.
        Moreover, the value corresponding to agent \( 0 \) in \( b \) remains unchanged (since \( v \neq 0 = b_0 \)).
        Finally, it was previously established that \( b \in B^* \) implies that it takes at least three different values.
        The transformation of replacing values for agents in \( A \), which in \( b \) have the same value \( v \), with \( x \),
        can decrease this value by at most one.
        Therefore, function \( b[x] \) attains at least two different values and is not the constant zero function.

        By the definition of \( f_{b} \), we have:
        \[
            \begin{split}
                f_{b}(v_{-1})  = \phi(b[v_{-1}]), \quad
                f_{b}(v)  = \phi(b), \quad
                f_{b}(v_1)  = \phi(b[v_1]).
            \end{split}
        \]
        Thus, from the inequalities assumed for contradiction, it follows that:
        \[
            f_{b}(v_{-1}) < f_{b}(v), f_{b}(v)>f_{b}(v_1),
        \]
        which implies that \( f_{b} \) is non-monotonic.
        To show a contradiction, we expand the definition of \( f_{b} \).
        By the expansion of \( \phi \), we find that \( f_{b} \) is a fraction.
        Let us denote the numerator of this fraction by \( g \) and the denominator by \( h \).
        By \cref{eq:phi-denominator}, we have:
        \[
            \begin{split}
                h(x) = 2\sum_{i \in N} |b[x]_i| =
                2\sum_{i \in A} x + 2\sum_{i \in N \setminus A} b_i =
                2|A|x + 2\sum_{i \in N \setminus A} b_i.
            \end{split}
        \]
        From the above form, it is evident that \( h(x) \) is a linear function of \( x \).
        Moreover, \( h(x) \) does not vanish for any \( x \in V \) since \( b[x] \) is not the constant zero function
        (as established previously, \( b[x] \) attains at least two distinct values).

        Now consider \( g(x) \), the numerator of the expression \( f_{b} \).
        By \cref{eq:phi-numerator}, we have:
        \[
            \begin{split}
                g(x) = \sum_j \frac{4|j|}{2k+1}|b[x]_j| +
                \sum_j|b[x]_j| +
                \sum_{j > 0} |b[x]_j| (2k+1 - 2j)(|b[x]_{j-k-1}| - |b[x]_{j-k}|) + \\
                \sum_{j < 0} |b[x]_j| (2k+1 + 2j)(|b[x]_{j+k+1}| - |b[x]_{j+k}|).
            \end{split}
        \]

        The term \( b[x]_i \) equals \( x \) for \( i \in A \) and \( b_i \) for \( i \notin A \).
        In the above expression, some terms are of the form \( b[x]_i b[x]_j \) for certain \( i, j \),
        which could potentially be quadratic with respect to \( x \).
        However, in every such term, the difference between \( i \) and \( j \) is at least \( k \).
        Thus, \( i \) and \( j \) have different signs or one of them is \( 0 \).
        Therefore, it cannot be the case that \( i \in A \) and \( j \in A \) simultaneously, as \( 0 \notin A \)
        and \( A \) is a subset of consecutive agents (because \( b \) is monotonic).
        This concludes the proof that all terms in the above expression for \( g(x) \) are linear in \( x \) or constant.
        Consequently, \( g(x) \) is a~linear function with respect to \( x \).

        Finally, we have \( f_{b}(x) = g(x)/h(x) \),
        where \( g(x) \) and \( h(x) \) are linear with respect to \( x \).
        Moreover, \( h(x) \) is non-zero for any \( x \in V \).
        By \cref{lem:monotonicznosc-funkcji-homograficznej}, this implies
        that the function \( f_{b} \) is monotonic over the entire domain \( V \).

        This contradiction completes the proof.
    \end{proof}

    \begin{proof}[Proof (of \cref{lem:apx-na-B*})]
        If function \( b^* \) is dominated, then \( \phi(b^*)\le 3/2-1/n<7/4 \) holds by \cref{lem:szacowanie-zdominowanych}.
        This completes the proof for dominated \( b^* \).

        Let us consider the case of non-dominated \( b^* \).
        To gain better insight into the structure of the function \( b^* \),
        let us represent it as a tuple \( (b^*_{-k},\ldots,b^*_k) \).
        By assumption that \( b^* \) is non-dominated at most \( k \) values in the tuple are equal to \( 0 \).
        By construction of every function in \( B^*\subsetneq B'' \) it also holds that \( b^*_0 = 0 \) and \( b^*_{-k} \le b^*_{-k+1} \le \ldots \le b^*_k \).
        Additionally, \( b^* \) is boundary, so it attains only values \( -0.5, 0, 0.5 \).
        Taking into account the above properties, tuple associated with \( b^* \) must be of the form \( (-0.5,\ldots,-0.5,0,\ldots,0,0.5,\ldots,0.5) \) where values \( -0.5, 0, 0.5 \) are repeated  \( k - m_{-1}, 1 + m_{-1} + m_1, k - m_1 \) times respectively,
        for some \( m_{-1},m_1\in[k] \).
        In other words:
        \begin{equation*}
            b^*(i)=\begin{cases}
                -0.5 & \text{if } i < -m_{-1},           \\
                0    & \text{if } -m_{-1} \le i \le m_1, \\
                0.5  & \text{if } m_1 < i.
            \end{cases}
        \end{equation*}

        Let us express the value of \( \phi(b^*) \) in terms of \( m_{-1},m_1,k \).
        We start by expanding the numerator of \( \phi \) for \( b^* \) as in \cref{eq:phi-numerator}:
        \begin{equation*}
            \begin{split}
                \sum_i \frac{4|i|}{2k+1}|b^*_i| +
                \sum_j|b^*_j| +
                \sum_{j > 0} |b^*_j| (2k+1 - 2j)(|b^*_{j-k-1}| - |b^*_{j-k}|) +
                \sum_{j < 0} |b^*_j| (2k+1 + 2j)(|b^*_{j+k+1}| - |b^*_{j+k}|).
            \end{split}
        \end{equation*}
        Let us consider the first sum in the above expression (the part corresponding to the \( RD \) mechanism).
        Based on the structure of \( b^* \) we have:
        \begin{equation*}
            \begin{split}
                \sum_i \frac{4|i|}{2k+1}|b^*_i| = \frac{2}{2k+1}\left(\sum_{i<-m_{-1}}|i| + \sum_{m_1 < i} |i|\right) =
                \frac{2}{2k+1}\frac{2k^2+2k-m_{-1}^2-m_{-1}-m_1^2-m_1}{2}.
            \end{split}
        \end{equation*}
        Second sum in the expression for the numerator of \( \phi(b^*) \) is equal to:
        \begin{equation*}
            \begin{split}
                \sum_j|b^*_j| = (k-m_{-1})0.5 + (k-m_1)0.5=  k-(m_{-1}+m_1)/2.
            \end{split}
        \end{equation*}
        Regarding the third sum, its terms depend on the difference between the values of \( b^* \) for two consecutive agents, \( j-k-1 \) and \( j-k \), for \( j > 0 \).
        For such range of indices, this difference is non-zero only for \( -m_{-1}=j-k \), i.e., \( j=k-m_{-1} \).
        Therefore, third sum in the expression for numerator of \( \phi(b^*) \) has only one non-zero term and is equal to:
        \begin{equation*}
            \begin{split}
                \sum_{j > 0} |b^*_j| (2k+1 - 2j)(|b^*_{j-k-1}| - |b^*_{j-k}|) =
                0.5(2k+1-2(k-m_{-1}))(|-0.5|-0) =
                (2m_{-1}+1)/4.
            \end{split}
        \end{equation*}
        Similar arguments can be used to show that the last sum in the expression for numerator of \( \phi(b^*) \) is equal to:
        \begin{equation*}
            \begin{split}
                \sum_{j < 0} |b^*_j| (2k+1 + 2j)(|b^*_{j+k+1}| - |b^*_{j+k}|) =
                (2m_1+1)/4.
            \end{split}
        \end{equation*}
        Summing up results above we obtain that the numerator of \( \phi(b^*) \) is equal to:
        \begin{equation*}
            \begin{split}
                \frac{2}{2k+1}\frac{2k^2+2k-m_{-1}^2-m_{-1}-m_1^2-m_1}{2}+
                k-(m_{-1}+m_1)/2 +
                (2m_{-1}+1)/4 +
                (2m_1+1)/4 = \\
                \frac{8k^2 + 8k - 2m_1^2 - 2m_1 - 2m_{-1}^2 - 2m_{-1} + 1}{2(2k + 1)}.
            \end{split}
        \end{equation*}

        Expanding the denominator of \( \phi(b^*) \) as in \cref{eq:phi-denominator} we obtain:
        \begin{equation*}
            \begin{split}
                2\sum_i |b^*_i| = 2\left(\sum_{i<-m_{-1}}|-0.5| + \sum_{m_1 < i} |0.5|\right) =
                2k-m_{-1}-m_1.
            \end{split}
        \end{equation*}

        Combining the results above, we finally obtain that:
        \[
            \begin{split}
                \phi(b^*) =
                \frac{8k^2 + 8k - 2m_1^2 - 2m_1 - 2m_{-1}^2 - 2m_{-1} + 1}{2(2k + 1)(2k - m_1 - m_{-1})} =\\
                \frac{8k^2 + 8k - 2(m_1 + m_{-1})^2 + 2m_1m_{-1} - 2(m_1 + m_{-1}) + 1}{2(2k + 1)(2k - (m_1 + m_{-1}))}
            \end{split}
        \]

        Now, let us apply the substitution \( m_1 := s - q, m_{-1} := s + q \) to the above expression.
        We obtain:
        \[
            \begin{split}
                \phi(b^*) =
                \frac{8k^2 + 8k - 8s^2 + 4s^2 - 4q^2 - 4s + 1}{2(2k + 1)(2k - 2s)}
                \leq \frac{8k^2 + 8k - 4s^2 - 4s + 1}{2(2k + 1)(2k - 2s)}=:R_{k,s}.
            \end{split}
        \]

        The last inequality holds because \( -4q^2 \) is non-positive. Thus, the numerator of the last fraction is at least as large as the numerator of the previous fraction.
        Regarding the denominator in the considered case, we assume that function \( b^* \) is not dominated. Hence, value \( 0 \) is attained no more than \( k \) times.
        That gives us the inequality \( m_{-1} + m_1 + 1 \leq k \).
        In terms of the new variables, \( s \) and \( q \), this is equivalent to \( s \leq (k - 1)/2 \).
        This implies that the denominator of the above expression is always positive, because it is a product of positive terms.

        The derivative of the \( R_{k,s} \) with respect to \( s \) is:
        \[
            \frac{\partial}{\partial s} R_{k,s} =
            \frac{8k^2 - 8ks + 4k + 4s^2 + 1}{4(k - s)^2(2k + 1)}.
        \]

        By previously showed inequality \( s \leq (k - 1)/2 \), the denominator of the above derivative is always positive, because it is a product of positive terms.
        Similarly, the numerator is positive because:
        \[
            8k^2 - 8ks = 8k(k - s) \geq 0,
        \]
        and the remaining terms \( 4s^2, 4k, 1 \) are clearly non-negative, with \( 1 \) being positive.
        Therefore, increasing \( s \) increases the value of the \( R_{k,s} \).

        From this, the following inequality holds (substituting the maximum allowable value for \( s \) i.e., \( (k - 1)/2 \)):
        \[
            \frac{8k^2 + 8k - 4s^2 - 4s + 1}{2(2k + 1)(2k - 2s)} \leq
            \frac{7k^2 + 8k + 2}{4k^2 + 6k + 2}.
        \]

        The right-hand side is a weighted average of the numbers \( 7/4, 8/6, 2/2 \) with non-negative weights \( 4k^2, 6k, 2 \).
        Thus, the value of this expression is no greater than:
        \[ \max(7/4, 8/6, 2/2) = 7/4. \]
        The chain of inequalities above completes the proof.
    \end{proof}

\fi

\end{document}